\documentclass[12pt]{article}

\usepackage{amsmath,amsgen,latexsym}
\usepackage{amstext,amssymb,amsfonts,amsthm}
\usepackage{cite}
\usepackage{multicol}
\usepackage{mathtools}

\usepackage{pstricks}   
\usepackage{hyperref}
\usepackage{color}
\usepackage[spanish,es-tabla,english]{babel}
\usepackage{graphicx}

\DeclarePairedDelimiter{\norm}{\lVert}{\rVert}

\newtheorem{lemma}{Lemma}
\newtheorem{theorem}{Theorem}
\newtheorem*{theorem*}{Theorem}
\theoremstyle{definition}
\newtheorem{definition}{Definition}

\author{Sergio Mercado\thanks{Email:
\texttt{somercado@pol.una.py}. Affiliation: Facultad Polit\'ecnica, Universidad Nacional de Asunci\'on, Campus Universitario, San Lorenzo C.P. 111421, Paraguay. This author is supported by a CONACyT scholarship for graduate studies.} \and Marcos Villagra\thanks{Email: \texttt{mvillagra@pol.una.py}. Affiliation: Departamento de Matem\'aticas, Universidad Nacional de Asunci\'on, Campus Universitario, San Lorenzo C.P. 111421, Paraguay.}}

\title{Bounds on the Spectral Sparsification of Symmetric and Off-Diagonal Nonnegative Real Matrices\thanks{The authors acknowledge the support of CONACyT research grants POSG17-62, PINV15-706 and PINV15-208.}}

\date{}

\begin{document}

\maketitle

\begin{abstract}
We say that a square real matrix $M$ is \emph{off-diagonal nonnegative} if and only if all entries outside its diagonal are nonnegative real numbers. In this note we show that for any off-diagonal nonnegative symmetric matrix $M$, there exists a nonnegative symmetric matrix $\widehat{M}$ which is sparse and close in spectrum to $M$.

\vspace{0.2cm}

\vspace{0.2cm}

\noindent\textbf{MSC Class.} 05C50; 68R01; 15B57; 15A42
\end{abstract}

\section{Introduction}

\subsection{Background}
The run-time of many important algorithms in mathematics and computer science depend on how ``sparse'' the input data is. One such example is Lanczsos's algorithm \cite{Lan50} which can be used to compute a set of eigenvalues and eigenvectors of a matrix of size $n$ in $\mathcal O(n^3)$ arithmetical operations in the worst-case. If the average number of non-zero entries per row in an input matrix to Lanczsos's algorithm is bounded by a constant, then its running-time can be bounded by $\mathcal O(n^2)$ arithmetic operations.

In this note, we show how to construct sparse matrices from a certain class of symmetric real matrices and present some potential applications with interesting research directions.

\subsection{Results}
The main approach of this work is to borrow some ideas from spectral graph theory in order to construct sparse matrices that are close in spectrum to what we call \emph{off-diagonal nonnegative} matrices.

\begin{definition}
    A real square matrix $M$ is {\it off-diagonal nonnegative} (or simply ODN matrix) if for all $i,j=1,\dots,n,$ with $i\neq j$ we have tha its entries $(M)_{ij}\geq 0.$
\end{definition}

Note that the diagonal elements of an ODN matrix can be any real number and only its off-diagonal elements are nonnegative real numbers.

Let $M$ be a symmetric ODN matrix. Define two matrices $A_M$ and $D_M$ as follows: (i) $(A_{M})_{ij} = (M)_{ij}$ if $i\neq j$ and $(A_{M})_{ij}=0$ if $i=j$, and (ii) $(D_{M})_{ij} = \sum_{\underset{j\neq i}{j\leq n}}(M)_{ij} $ if $j = i$, and $(D_{M})_{ij} =0$ if $j \neq i$. Define a third matrix $L_M=D_M-A_M$. Note that the matrices $A_M$ and $L_M$, as constructed, can be respectively interpreted as the adjacency and Laplacian matrices of some graph.

Let us denote by $\Delta_{M}$ and $\delta_{M}$ the largest and smallest element in the diagonal of $M$, respectively. Now we are ready to state our main result.

\begin{theorem}\label{the:main}
    Let $M \in \mathbb R^{n\times n}$ be an ODN symmetric matrix with $m$ non-zero off-diagonal entries and a set of  eigenvalues $\lambda_1\geq\cdots \geq \lambda_n$ with a respective set of eigenvectors $x_{1}, x_{2}, \cdots , x_{n}$. For any $\epsilon>0$ such that $0<\epsilon \leq 1/120$, there exists an ODN symmetric matrix $\widehat{M}$ $\in \mathbb R^{n\times n}$ with $\mathcal O(\frac{n}{\epsilon^2})$ non-zero entries, with a set of eigenvalues $\widehat{\lambda}_{1}\geq \cdots \geq \widehat{\lambda}_{n}$ with a respective set of eigenvectors $\widehat{x}_{1}, \cdots, \widehat{x}_{n}$, such that
    \[|\lambda_{i}-\widehat{\lambda}_{i}| \leq\epsilon\sqrt{n} \rho(L_{M})+\frac{\Delta_{M}-\delta_{M}}{2}.
    \]
    Furthermore, if $\theta_i$ is the angle between the subspaces spanned by eigenvectors $x_i$ and $\widehat{x}_i$, then 
    \[
    \sin\theta_{i} \leq \frac{\epsilon\sqrt{n}\rho(L_{M})+(\Delta_{M}-\delta_{M})/2}{\min\{|\widehat{\lambda}_{i-1}-\lambda_{i}|, \ |\lambda_{i}-\widehat{\lambda}_{i+1}|\}}.
    \]
\end{theorem}

Note that if all diagonal entries in $M$ are close in values, then $(\Delta_{M}-\delta_{M})/2$ is small which will give a bound of $\epsilon\sqrt{n} \rho(L_{M})$.

Theorem \ref{the:main} is comparable to a result of Zouzias \cite{Zou12}. Given $0<\epsilon<1$, for any self-adjoint matrix $A$ of size $n$ that is also $\theta$-symmetric diagonally dominant, there exists a matrix $\widehat{A}$ with at most $\mathcal O(n\theta \log n/\epsilon^2)$ non-zero entries\footnote{The ``$\log n$'' factor can be dropped using new spectral sparsification techniques like Lee and Sun \cite{LS18}} and $\norm{A-\widehat{A}}\leq \epsilon\norm{A}$. Informally, a matrix $A$ is $\theta$-symmetric diagonally dominant if $\norm{A}_\infty = \mathcal O(\sqrt{\theta})$. Therefore, the approximation factor in the result of Zouzias \cite{Zou12} has a dependency on the entry with the largest absolute value in the matrix $A$. Theorem \ref{the:main} eliminates that dependency at the expense of further restrictions on the matrix we want to approximate.

\subsection{Some Applications of Theorem \ref{the:main}}

\subsubsection{Optimization of Quadratic Forms}
Let $A\in \mathbb R^{n\times n}$ be an ODN matrix. We say that a quadratic form $Q(x)=x^T Ax$ is ODN if its matrix $A$ is ODN. Every quadratic form has a diagonal form $Q(x)=\lambda_1 x_1^2+\cdots \lambda_n x_n^2$, where each $\lambda_i$ is an eigenvalue of $A$. Furthermore, Sylvester's Law of Inertia tells us that the number of positive and negative coefficients in any diagonal form of $Q$ is an invariant of $Q$.  

Let $\widehat{Q}(x)=x^T \widehat{A} x$ where $\widehat{A}$ is a matrix obtained from $A$ by means of Theorem \ref{the:main}. If $\widehat Q(x)=\widehat \lambda_1 x_1^2+\cdots \widehat \lambda_n x_n^2$ where $\widehat \lambda_i$ is an eigenvalue of $\widehat A$, we have that $|Q(x)-\widehat Q(x)|\leq \rho(A)+(\Delta_{A}-\delta_{A})/2$ by choosing $\epsilon$ sufficiently small. Thus, if we are interested in the optimization of $Q(x)$, we can use $\widehat A$ as the input of a quadratic optimization solver instead of $A$, and it will result in a solution with the guarantees mentioned in the previous sentence. State-of-art quadratic optimization solvers can exploit the sparsity of an input matrix, which is especially important in the case of non-convex problems.

Optimization of quadratic forms is an NP-hard optimization problem, even with binary variables. In fact, a single negative eigenvalue suffices to make the problem NP-hard \cite{PV91}. It is also closely related to the optimization of the Ising model in statistical mechanics. We do not know, however, if the optimization of ODN quadratic forms is NP-hard or not, and we leave this as an open problem.

\subsubsection{Principal Component Analysis}

Principal Componnet Analysis or {\it PCA} is a method to reduce the dimensionality of data \cite{Jol02}. Given a set of data with correlated atributes, PCA reduces the number of atributes while it preserves the variance, as best possible, of the data. PCA constructs a new set of non-correlated variables known as \emph{principal components}.

Let $x$ be a vector of $n$ random variables. The first principal component is defined as $z_1=v_1^Tx$, where $v_1\in \mathbb{R}^n$ is a vector that maximizes the variance of $z_1$, denoted $Var(z_1)$. The $i$-th principal component is the variable $z_i=v_i^t x$, with $v_i\in \mathbb R^n$, such that each $z_1,z_2,\dots,z_i$ are pairwise uncorrelated and $Var(z_i)$ is maximum. It is known that if $S$ is the covariance matrix of $x$, with eigenvalues $\lambda_1\geq \cdots \lambda_n$, then $v_i$ is the eigenvector of $S$ corresponding to $\lambda_i$ and $Var(z_i)=\lambda_i$.

Another approach to PCA is to use a correlation matrix $M$ instead of the covariance matrix $S$. The matrix $M$ is the covariance matrix of a vector $x$ that is normalized by subtracting its mean from each entry and then divided by its standard deviation. If $X$ is a data matrix of size $m\times n$, with no loss of generality, suposse that $M=(1/n)X^T X$. Suppose further that $M$ is ODN. By Theorem \ref{the:main} we can obtain a sparse matrix $\widehat{M}$ that is close in spectrum to $M$. Let $z_i$ and $\widehat{z}_i$ be the $i$-th principal components of $M$ and $\widehat M$, respectively. Thus, for the first principal component we have
\[
 Var(z_{1}) - \epsilon\sqrt{n}\rho(L_{M}) \leq  Var(\widehat{z}_{1}) \leq Var(z_{1}) + \epsilon\sqrt{n}\rho(L_{M}).
\]
In general, for the first $p$ principal components we have that
\[
 \sum_{i = 1}^{p}Var(z_{i}) - i\epsilon\sqrt{n}\rho(L_{M}) \leq \sum_{i = 1}^{p} Var(\widehat{z}_{i}) \leq \sum_{i = 1}^{p}Var(z_{i}) + i\epsilon\sqrt{n}\rho(L_{M}). 
\]

This loss in precision is compensated with a gain in speed-up on the computation of eigenvalues, which is a crucial step in PCA. For example, if we use Lanczos's algorithm \cite{Lan50}, which is sensitive to the sparsity of its input matrix, we can compute eigenvalues and an orthonormal set of eigenvectors using $\mathcal O(n^2/\epsilon^2)$ arithmetic operations, if we assume $\mathcal O(1/\epsilon^2)$ non-zero entries per row on average.

\subsection{Outline}
The remaining of this note is organized as follows. Section \ref{sec:preliminaries} presents the notation on linear algebra used throughout this work and briefly reviews the main concepts and techniques on spectral sparsification of graphs. Section \ref{sec:lemmas} presents some technical lemmas that are necessary for our proof of Theorem \ref{the:main}. Finally, Section \ref{sec:main} presents a full proof of Theorem \ref{the:main}.

\section{Preliminaries}\label{sec:preliminaries}
In this section we introduce the notation used throughout this work and present some basic facts from linear algebra and spectral sparsification of graphs. In the entirety of this work we use $\mathbb R$ to denote the set of real numbers and $\mathbb R^+$ to denote the set of positive real numbers.

\subsection{Linear Algebra}
Let $M$ be a real matrix. We use $(M)_{ij}$ to denote the element in the $i$-th row and $j$-th column of $M$. Let $\norm{M}$, $\norm{M}_{\infty}$ and $\norm{M}_{1}$ denote the 2-norm, the $\infty$-norm and the 1-norm, respectively. The spectral radius of $M$ is denoted as $\rho(M)$. Recall that if $M$ is symmetric, then $\norm{M}=\rho(M)$.

The following theorem characterizes the eigenvalues of real symmetric matrices.

\begin{theorem*}[Courant-Fisher]\label{Teo:Cou-Fisch}
    Let $M \in \mathbb{R}^{n\times n}$ be a symmetric matrix with eigenvalues $\alpha_{1} \geq \cdots \geq \alpha_{n}.$ Then 
    \[
    \alpha_{k} = \max_{\underset{dim(S) = k }{S \subset \mathbb{R}^{n}}} \min_{\underset{x \neq 0}{x\in S}}\frac{x^{T}Mx}{x^{T}x} = \min_{\underset{dim(T) = n-k+1}{T\subset \mathbb{R}^{n}}} \max_{\underset{x\neq 0}{x\in T}}\frac{x^{T}Mx}{x^{T}x},
    \]
    where the maximization and minimization are over subspaces $S$ and $T$ of $\mathbb{R}^{n}.$
\end{theorem*}

The following theorem due to Davis and Kahan \cite{DK70} is important in perturbation theory. The simplified version of the theorem we present here is by Yu, Wang and Samworth \cite{YWS15}.

\begin{theorem*}[Davis-Kahan \cite{DK70}]
    Let $A$ and $B$ be symmetric matrices, and  $R = A-B$. Let $\alpha_{1} \geq \cdots \geq \alpha_{n}$ be the eigenvalues of $A$ with corresponding eigenvectors $a_{1}, a_{2}, ..., a_{n}$, and let $\beta_{1} \geq \cdots \geq \beta_{n}$ be the eigenvalues of $B$ with corresponding eigenvectors $b_{1}, b_{2}, ... ,b_{n}.$ Let $\theta_{i}$ be the angle between $a_{i}$ and $b_{i}$, then
    \[
    \sin\theta_{i} \leq \frac{ \norm{R} }{\min \{|\beta_{i-1} - \alpha_{i}|, \ |\beta_{i+1}-\alpha_{i}|\}}.
    \] 
\end{theorem*}

The last theorem gives a bound between eigenvalues of two symmetric matrices $A$ and $B$ from the norm $\norm{A-B}$; see the textbook of Bhatia \cite{Bha97}, page 63, for a proof.

\begin{theorem*}[Weyl's Perturbation Theorem]\label{teo:weyl} 
    Let $A, B \in\mathbb{R}^{n\times n}$ be symmetric matrices such that $\alpha_{1}, \alpha_{2},..., \alpha_{n}$, and $\beta_{1}, \beta_{2},..., \beta_{n}$ are the eigenvalues of $A$ and $B$, respectively. Then, for all $i=1,..., n$, 
    \[
    |\alpha_i - \beta_i|\leq \norm{A-B}.
    \]
\end{theorem*}

\subsection{Spectral Sparsification of Graphs}
Let $G=(V, E, w)$ be a simple undirected graph with vertices $v_{1}, v_{2},..., v_{n}\in V$. An edge in $E$ between vertices $v_i$ and $v_j$ is denoted $ij$. Each edge $ij$ has a weight assigned to it according to a weight function $w:E\to \mathbb{R}^+\cup \{0\}$. As a short-hand we use $w_{ij}=w(ij)$. The adjacency matrix $A_{G}$ of $G$ is defined as $(A_{G})_{ij}=w_{ij}$ if  $i \neq j$, and $(A_{G})_{ij}=0$, if $i=j$. We also define the degree matrix $D_{G}$ of $G$ as $(D_{G})_{ij}=\sum_{j=1}^{n}w_{ij}$ if $j=i$, and $(D_{G})_{ij}=0$ otherwise. The Laplacian matrix $L_G$ of $G$ is defined as $L_G=D_G-A_G$.

Spectral sparsification is a method introduced by Spielman and Teng \cite{ST11} that is used to construct a sparse graph $\widehat{G}$ from any given graph $G$ such that the spectrum of their Laplacian matrices are ``close.'' For any $\epsilon\in (0,1)$, we say that $\widehat{G}$ is an $\epsilon$-spectral sparsifier of $G$ if for every $x\in \mathbb R^n$ it holds that
\begin{equation}\label{eq:quadratic-form}
(1-\epsilon)x^{T}L_{G}x \leq x L_{\widehat{G}}x \leq (1+ \epsilon)x^{T}L_{G}x.
\end{equation}
It is clear that if the eigenvalues of $L_{\widehat{G}}$ are $\widehat{\mu}_{1} \geq  \cdots \geq \widehat{\mu}_{n}$, then by Eq.(\ref{eq:quadratic-form}) and the Courant-Fischer theorem
    \[
    (1-\epsilon)\mu_{i} \leq \widehat{\mu_{i}} \leq (1+\epsilon)\mu_{i}
    \]
for all $i=1,..., n$.

For any two square matrices $A,B$ we use $A \preceq B$ whenever $B-A$ is positive semidefinite. Thus, we can succinctly write Eq.(\ref{eq:quadratic-form}) as
\begin{equation}
(1-\epsilon)L_G \preceq L_{\widehat{G}} \preceq (1+\epsilon)L_G.
\end{equation}

Spielman and Teng \cite{ST11} proved that every graph with positive weights has an $\epsilon$-spectral sparsifier close to linear-size in the number of vertices. The following theorem is  currently the best construction of spectral sparsifiers.

\begin{theorem*}[Lee-Sun \cite{LS18}]\label{lee-sun}
    Given any integer $q\geq 10$ and $0<\epsilon \leq 1/120$. Let $G=(V,E,w)$ be an undirected and weighted graph with $n$ vertices and $m$ edges. Then, there exists a $(1+\epsilon)$-spectral sparsifier of $G$ with $\mathcal O(\frac{qn}{\epsilon^2})$ edges.
\end{theorem*}

\section{Tecnical Lemmas}\label{sec:lemmas}
In this section we show some technical lemmas that will help us in proving Theorem \ref{the:main}.

\begin{lemma}\label{Lema1}
    Let $G$ be a graph of $n$ vertices and $\widehat{G}$ be an $\epsilon$-spectral sparsifier of $G$. Let $\mu_1, \mu_2,\cdots , \mu_n$ be the eigenvalues of $L_G$ with respective eigenvectors $x_1, x_2 \cdots,x_n$ and let $\widehat{\mu}_1, \widehat{\mu}_2, \cdots, \widehat{\mu}_{n}$ be the eigenvalues of $L_{\widehat{G}}$ with respective eigenvectors $\widehat{x}_{1}, x_{2},\cdots \widehat{x}_{n}.$ Then
    \begin{enumerate}
    \item $\norm{ L_G - L_{\widehat{G}}} \leq \epsilon\rho(L_G)$, and
    \item if $\theta_{i}$ is the angle between $x_{i}$ and $\widehat{x}_{i}$, then 
    \[
    \sin \theta_{i} \leq \frac{\epsilon \rho(L_G)}{\min\{|\mu_{i} - \widehat{\mu}_{i-1}|, |\mu_{i}-\widehat{\mu}_{i+1}|\}},
    \]
    where we assume $|\mu_i-\widehat \mu_{i\pm 1}|\neq 0$ for all $i=1,\dots,n$.
    \end{enumerate}
\end{lemma}
\begin{proof}
    Since $\widehat G$ is an $\epsilon$-spectral sparsifier of $G$ we have that
    \[
    (1- \epsilon)L_{G} \preceq {L}_{\widehat G} \preceq (1+\epsilon) L_{G},
    \]
    which implies
    \begin{equation}\label{eq:ineq}
    L_{G} - {L}_{\widehat{G}} \preceq \epsilon L_{G}.
    \end{equation}
    Note that $\delta_i$ is an eigenvalue of $L_G-L_{\widehat{G}}$ if and only if $-\delta_i$ is an eigenvalue of $L_{\widehat{G}}-L_G$. With no loss of generality suppose that $\rho(L_G-L_{\widehat{G}})$ coincides with the largest eigenvalue of $L_G-L_{\widehat{G}}$ and let $z$ be a normalized eigenvector associated to $\rho(L_G-L_{\widehat G})$. Then
    \begin{align*}
    \norm{  L_{G} - L_{\widehat{G}} }
        &= \rho(L_{G} - {L}_{\widehat{G}}) && (L_G-L_{\widehat{G}}\text{ is symmetric})\\
        &= z^{T}(L_{G} - {L}_{\widehat{G}})z\\
        &\leq z^{T}(\epsilon L_{G})z && (\text{from Eq.(\ref{eq:ineq}))}\\
        &\leq \rho(\epsilon L_G)\\
        &= \epsilon \cdot\norm{L_{G}}=\epsilon\cdot\rho(L_{G}).
    \end{align*}
     The second part of the lemma is implied by the Davis-Kahan theorem and the first part of this lemma, thus completing the proof.
\end{proof}


 
\begin{lemma}\label{lem:lema3}
    Let $L_{G}=D_{G}-A_{G}$ y $L_{H}=D_{H}-A_{H}$ be the Laplacian matrices of graphs $G$ and $H$, respectively. Then
    \[
    \norm{A_{G} - A_{H}} \leq \sqrt{n}\norm{L_{G}-L_{H}}.
    \]
\end{lemma}
\begin{proof}
The matrix $A_{G}-A_{H}$ is symmetric, and hence, $\norm{A_{G}-A_{H}}_{\infty}=\norm{A_{G}-A_{H}}_{1}$. Then, using the inequality of \cite[Th.2.11-5]{Stewart-Sun}\footnote{$\norm{A}^{2}\leq \norm{A}_1\cdot \norm{A}_\infty$ for any matrix $A$.} we have that $\norm{A_{G}-A_{H}}\leq \norm{A_{G}-A_{H}}_{\infty}$. On the other hand, $\norm{L_{G}-L_{H}}_{\infty}=\norm{A_{G}-A_{H}}_{\infty} + \underset{i\leq n}{\max}\big|(D_{G})_{ii} - (D_{H})_{ii}  \big|$, then
\begin{align*}
\norm{A_G-A_{H}}_{\infty} &\leq \norm{L_{G}-L_{H}}_{\infty}\\
    &\leq\sqrt{n}\norm{L_{G}-L_{H}},
\end{align*}
and the lemma thus follows.
\end{proof}
\section{Proof of the Theorem \ref{the:main}}\label{sec:main}
Let $M$ be an ODN symmetric matrix and let $\overline{M}$ be another matrix defined as
\[
(\overline{M})_{ij}= \left\{ \begin{array}{lcc}
             (M)_{ij} &  if & i\neq j \\
            \\  \ d &  if & i=j,
             \end{array}
   \right.
\]
where $d=(\Delta_{M}+\delta_{M})/2$.
Then we have that $\norm{M-\overline{M}}=\max\{\Delta_{M}-d, \delta_{M}-d\} = (\Delta_{M}-\delta_{M})/2$, and thus, the chosen $d$ is the value that minimizes the norm $\norm{M-\overline{M}}$.

Now let $\lambda_{1}\geq\cdots\geq\lambda_{n}$ and $\overline{\lambda}_{i}\geq\cdots\geq\overline{\lambda}_{n}$ be the eigenvalues of $M$ and $\overline{M}$, respectively. By our definition of $\overline{M}$ and Weyl's Perturbation Theorem we have that
    \begin{equation}\label{iq:t1}
	    |\lambda_{i}-\overline{\lambda}_{i}|\leq \frac{\Delta_{M}-\delta_{M}}{2}.
    \end{equation}

Recall that if we define two matrices $A_M$ and $D_M$ where $(A_{M})_{ij} = (M)_{ij}$ if $i\neq j$, and $(A_{M})_{ij}=0$ for $i=j$, and $(D_{M})_{ij} = \sum_{\underset{j\neq i}{j\leq n}}(M)_{ij} $ if $j = i$, and $(D_{M})_{ij} =0$ if $j \neq i$, we can see $A_M$ and $D_M$ as the adjacency and degree matrices of some graph $G_M$. Consequently we have a Laplacian matrix $L_{M} = D_{M} - A_{M}$. Thus, notice that $L_{M}=L_{\overline{M}}$, where analogously we define $L_{\overline{M}}=D_{\overline{M}}-A_{\overline{M}}$.

By the Lee-Sun theorem we know that given $\epsilon$ with $0<\epsilon<1/120$, there exists a matrix $\widehat{L}_{M}$ with $\mathcal O(n/\epsilon^{2})$ non-zero entries that is an $\epsilon$-spectral sparsifier of $L_{M}$. Then by Lemma \ref{Lema1} we have that 
\[
\norm{L_{M}-\widehat{L}_{M}}\leq \epsilon\rho(L_{M}),\]
and hence, by Lemma  \ref{lem:lema3} it follows that
\[
\norm{A_{M}-\widehat{A}_{M}}\leq\epsilon\sqrt{n}\rho(L_{M}).
\]
If we let $\widehat{M}=\widehat{A}_{M}+dI$, then from the last inequality we have 
\[
\norm{\overline{M}-\widehat{M}}\leq\epsilon\sqrt{n}\rho(L_{M}).
\]
Thus, if $\widehat{\lambda}_{1}\geq\widehat{\lambda}_{2}\geq\cdots\geq\widehat{\lambda}_{n} $ are the eigenvalues of $\widehat{M}$, by Weyl's Perturbation Theorem, we can see that
\begin{equation}\label{iq: overline_lambda }
|\overline{\lambda}_{i}-\widehat{\lambda}_{i}|\leq\epsilon\sqrt{n}\rho(L_{M}).
\end{equation}

Therefore, from equations (\ref{iq:t1}) and (\ref{iq: overline_lambda }) it follows that
    \begin{align*}
    |\lambda_{i}-\widehat{\lambda}_{i}|&=|\lambda_{i}-\overline{\lambda}_{i}+\overline{\lambda}_{i}-\widehat{\lambda}_{i}|\\
    &\leq |\lambda_{i}-\overline{\lambda}_{i}|+|\overline{\lambda}_{i}-\widehat{\lambda}_{i}|\\
    &\leq\epsilon\sqrt{n}\rho(L_{M})+\frac{\Delta_{M}-\delta_{M}}{2}.
    \end{align*}

The fact that $\widehat{M}$ has $\mathcal O(\frac{n}{\epsilon^{2}})$ non-zero elements follow directly from the Lee-Sun theorem. Finally, the last part of the theorem is obtained by means of the Davis-Kahan theorem
    \[
    \sin\theta_{i}\leq \frac{\norm{M-\widehat{M}}}{\min\{|\widehat{\lambda}_{i-1}-\lambda_{i}|, \ |\lambda_{i}-\widehat{\lambda}_{i+1}|\}} \leq \frac{\epsilon\sqrt{n}\rho(L_{M})+(\Delta_{M}-\delta_{M})/2}{\min\{|\widehat{\lambda}_{i-1}-\lambda_{i}|, \ |\lambda_{i}-\widehat{\lambda}_{i+1}|\}}.
    \]

\end{document}